\documentclass[journal]{IEEEtran}
\usepackage{amsmath,amssymb,graphicx,psfrag,cite}
\usepackage{amsfonts}
\usepackage{color}
\usepackage{srcltx}
\usepackage{comment}
\renewcommand{\QED}{\QEDopen}

\newcommand{\cM}{{\mathcal{M}}}

\def\bu{\bar{u}}

\newcommand{\sff} {\mathsf{f}}

\newtheorem{theorem}{Theorem}
\newtheorem{definition}{Definition}
\newtheorem{corollary}{Corollary}

\newtheorem{lemma}{Lemma}
\newtheorem{example}{Example}

\author{Enrico Paolini, Mark F. Flanagan, Marco Chiani and Marc P. C. Fossorier
\thanks{E. Paolini and M. Chiani are with DEIS/WiLAB, University of Bologna, via Venezia 52, 47023 Cesena (FC), Cesena, Italy (e-mail: e.paolini@unibo.it, marco.chiani@unibo.it).}%
\thanks{M. F. Flanagan is with the School of Electrical, Electronic and Mechanical Engineering, University College Dublin, Belfield, Dublin 4, Ireland (e-mail: mark.flanagan@ieee.org).}%
\thanks{M. P. C. Fossorier is with ETIS ENSEA, UCP, CNRS UMR-8051, 6 avenue du Ponceau, 95014 Cergy Pontoise, France (e-mail:  mfossorier@ieee.org).}
}

\begin{document}


\title{Spectral Shape of Check-Hybrid GLDPC Codes}

\maketitle
\thispagestyle{empty}
\pagestyle{empty}

\begin{abstract}
This paper analyzes the asymptotic exponent of both the weight spectrum and the stopping set size spectrum for a class of generalized low-density parity-check (GLDPC) codes. Specifically, all variable nodes (VNs) are assumed to have the same degree (regular VN set), while the check node (CN) set is assumed to be composed of a mixture of different linear block codes (hybrid CN set). A simple expression for the exponent (which is also referred to as the \emph{growth rate} or the \emph{spectral shape}) is developed. This expression is consistent with previous results, including the case where the normalized weight or stopping set size tends to zero. Furthermore, it is shown how certain symmetry properties of the local weight distribution at the CNs induce a symmetry in the overall weight spectral shape function.
\end{abstract}

\section{Introduction}
Tanner codes were introduced in \cite{Tanner_GLDPC} as a generalization of Gallager's low-density parity-check (LDPC) codes \cite{Gallager}. In the bipartite graph representation of a Tanner code, all variable nodes (VNs) have the same degree and may be interpreted as repetition codes with the same length (regular VN set). Moreover, all check nodes (CNs) are generic linear block codes with the same length, dimension and code book (regular CN set). Tanner codes with an \emph{irregular} VN set (i.e., VNs with different degrees) or with a \emph{hybrid} CN set (i.e., CNs of different types) are more generically referred to as generalized LDPC (GLDPC) codes in the literature (e.g. \cite{Liva2008:Tanner}). Note that an LDPC code may be viewed as a GLDPC code where all CNs are single parity-check (SPC) codes. An even more general class of codes is represented by doubly-generalized LDPC (D-GLDPC) codes \cite{Wang_Fossorier_DG_LDPC}, where the VNs are also allowed to be generic linear block codes.

In this paper, a simple formula for the asymptotic (in codeword length) exponent of the weight distribution of a GLDPC code ensemble with a regular VN set and a hybrid CN set is developed. As usual in the literature, this exponent will be referred to as the \emph{growth rate of the weight distribution} or the \emph{weight spectral shape} of the ensemble, the two expressions being used interchangeably throughout the paper. The starting point for deriving the above-mentioned formula is a polynomial system solution for the spectral shape that was developed by the authors in \cite[Theorem~1]{Flanagan09:weight}. Here, it was shown that any value of the spectral shape function of an irregular D-GLDPC code ensemble can be calculated by solving a $(4 \times 4)$ system of polynomial equations, regardless of the number of different VN and CN~types. Additional recent works relevant to the subject of this paper are \cite{abusurra2007:enumerators,declercq2008:nonbinary,Kasai2009:multiedge,lentmaier2009:exact,wang2009:multiedge}.

As explained in Section~\ref{section:preliminaries}, assuming transmission over the binary erasure channel (BEC) and iterative decoding, the developed formula is also valid for the asymptotic exponent of the stopping set size distribution upon replacing the local weight enumerating function (WEF) of each CN type with an appropriate polynomial function.

Symmetry properties of the growth rate of the weight distribution are also investigated. It is proved that the weight spectral shape function of a variable-regular GLDPC ensemble is symmetric w.r.t. normalized weight $\alpha = 1/2$ if the local WEF of each CN is a symmetric polynomial. This result establishes a connection between symmetry properties at a ``microscopic'' level (i.e., at the nodes of the Tanner graph) and symmetry of the ``macroscopic'' growth rate function. A necessary condition for symmetry of the weight spectral shape is also developed.

\section{Preliminary Definitions}\label{section:preliminaries}
We consider a check-hybrid GLDPC code ensemble $\cM_n$, where $n$ is the codeword length (this is
equal to the number of VNs). All VNs are repetition codes of length $q \geq 2$, with input-output
weight enumerating function   
\begin{align}\label{eq:B_repetition}
B(x,y)=1+xy^q \, .
\end{align} 
There are $n_c$ different CN types $t \in I_c = \{ 1,2,\cdots, n_c\}$; for each CN type $t \in I_c$,
we denote by $h_t$, $s_t$ and $r_t$ the CN dimension, length and minimum distance, respectively.
We assume that $r_t \ge 2$ for all $t \in I_c$, and that no CN has idle bits (i.e., its generator
matrix contains no all-zero column). The WEF for CN type $t \in I_c$ is given by 
\begin{eqnarray*}
A^{(t)}(z) & = & \sum_{u=0}^{s_t} A_u^{(t)} z^u = 1 + \sum_{u=r_t}^{s_t} A_u^{(t)} z^u \; .
\end{eqnarray*}
Here $A_u^{(t)} \ge 0$ denotes the number of weight-$u$ codewords for CNs of type $t$. We denote by $\bar{u}_t$ the largest $u \in \{r_t,r_t+1,\dots,s_t\}$ such that $A^{(t)}_u > 0$. 

For $t \in I_c$, $\rho_t$ is the fraction of edges of the Tanner graph connected to type-$t$ CNs,
and the polynomial $\rho(x)$ is defined~by
\[
\rho(x) = \sum_{t\in I_c} \rho_t x^{s_t - 1} \; .
\]
If $E$ denotes the number of edges in the Tanner graph, the number of CNs of type $t\in I_c$ is then
given by $E \rho_t / s_t$. Denoting as usual $\int_0^1 \rho(x) \, {\rm d} x$ by $\int\! \rho$, we
see that the number of edges in the Tanner graph is given by $E = n q$ and the number of CNs is
given by $m = E \int\! \rho$. Therefore, the fraction of CNs of type $t \in I_c$ is given by
\begin{equation}
\gamma_t = \frac{\rho_t}{s_t \int\! \rho} \; .
\label{eq:gamma_t_definition}
\end{equation} 
A member of the GLDPC code ensemble $\cM_n$ corresponds to a permutation of the $E$ edges connecting VNs to CNs.

In the special case where there is only \emph{one} CN type (Tanner code ensemble), we write the WEF for this CN type as $A(z) = \sum_{u=0}^{s} A_u z^u = 1 + \sum_{u=r}^{s} A_u z^u$. In this case, the largest $u \in \{r,r+1,\dots,s\}$ such that $A_u > 0$ is denoted by $\bar{u}$. 

The weight spectral shape of a GLDPC code ensemble sequence $\{ \cM_n \}$ is defined by 
\begin{equation}
G(\alpha) \triangleq \lim_{n\rightarrow \infty} \frac{1}{n} \log \mathbb{E}_{\cM_n} \left[ N_{\alpha n} \right]
\label{eq:growth_rate_result}
\end{equation}
where $\mathbb{E}_{\cM_n}$ denotes the expectation operator over the ensemble $\cM_n$, $N_{w}$
denotes the number of codewords of weight $w$ of a randomly chosen GLDPC code in the ensemble and
the logarithm has base $e$. The limit in (\ref{eq:growth_rate_result}) assumes the inclusion of only
those positive integers $n$ for which $\alpha n \in \mathbb{Z}$ and $\mathbb{E}_{\cM_n} [ N_{\alpha
n} ]$ is positive. Using standard notation, we also define the asymptotic relative minimum distance
for $\cM_n$as $\alpha^*=\inf \{ \alpha > 0 \; | \; G(\alpha)\geq 0 \}$. The ensemble sequence is
said to exhibit \emph{good spectral shape behavior} when $\alpha^* > 0$ and \emph{bad spectral
shape behavior} when $\alpha^* = 0$.

Although this paper is focused on the weight spectrum, the results developed in Section \ref{section:spectral_shape_formula} can be extended to the stopping set size spectrum. A stopping set of a GLDPC code may be defined as any subset $\mathcal{S}$ of the VNs such that, assuming all VNs in $\mathcal{S}$ are erased and all VNs not in $\mathcal{S}$ are not erased, every CN which is connected to $\mathcal{S}$ cannot recover any VN in $\mathcal{S}$.\footnote{The concept of stopping set was first introduced in \cite{di02:finite} in the context of LDPC codes. When applied to LDPC codes (i.e., all CNs are SPC codes), the definition of stopping set used in this paper coincides with that in \cite{di02:finite}.} A \emph{local stopping set} for a CN is a subset of the local code bits which, if erased, is not recoverable to any extent by the CN. All results derived in this paper for the distance spectrum can be extended to the stopping set size spectrum by simply replacing the WEF for CN type $t \in I_c$ with its local stopping set enumerating function (SSEF).

We point out that the local SSEF of a CN depends on the decoding algorithm used to locally recover from erasures. In this paper, we will consider both bounded distance (BD) and maximum \emph{a posteriori} (MAP) CN decoding. In the former case, the local SSEF (BD-SSEF) is given by 
\begin{align}\label{eq:Psi_def}
\Psi^{(t)}(z) = 1 + \sum_{u=r_t}^{s_t} {s_t \choose u} z^u \; .
\end{align}
In the latter case, the local SSEF (MAP-SSEF) is given by
\begin{align}\label{eq:Phi_def}
\Phi^{(t)}(z) = 1 + \sum_{u=r_t}^{s_t} \phi_u^{(t)}\, z^u
\end{align}
where $\phi^{(t)}_u \geq 0$ is the number of local stopping sets (under MAP decoding) of size $u$.\footnote{Denoting by $\mathbf{G}_t$ any generator matrix for a type-$t$ CN, a local erasure pattern is a local stopping set under MAP decoding when each column of $\mathbf{G}_t$ corresponding to erased bits is linearly independent of the columns of $\mathbf{G}_t$ corresponding to the non-erased bits.}

The growth rate of the stopping set size distribution of the ensemble sequence $\{ \cM_n \}$ for the
case of BD and MAP decoding at the CNs, whose definition is analogous to
\eqref{eq:growth_rate_result}, will be denoted by $G_{\Psi}(\alpha)$ and $G_{\Phi}(\alpha)$,
respectively. Similarly the asymptotic relative minimum stopping set size will be denoted by
$\alpha_{\Psi}^*$ and $\alpha_{\Phi}^*$, respectively.

\begin{definition} Let 
\begin{equation}
M \triangleq \left( \int\! \rho \right) \sum_{t \in I_c} \gamma_t \bar{u}_t \leq 1
\label{eq:M_definition}
\end{equation}
and define the function $\sff: \mathbb{R}^+ \rightarrow [0,M)$ as
\begin{align}\label{eq:f}
\sff(z) = \left( \int\! \rho \right) \sum_{t \in I_c} \gamma_t \frac{ z \, \frac{\mathrm{d} A^{(t)}(z)}{\mathrm{d} z}}{A^{(t)}(z)} \, .
\end{align}
\end{definition}

Note that we have $M=1$ if and only if $\bu_t=s_t$ for all $t \in I_c$.
\begin{lemma} 
The function $\sff$ fulfills the following properties:
\begin{enumerate}
\item It is monotonically increasing for all $z>0$;
\item $\sff(0) = \sff'(0) = 0$;
\item $\lim_{z\rightarrow +\infty} \sff(z) = M$.
\end{enumerate}
\end{lemma}
\begin{proof} We prove the first property, as the proofs of the second and the third properties are straightforward. The derivative of $\sff$ (normalized w.r.t. $\int\! \rho$) is given by
\begin{align*} 
\sum_{t \in I_c} \gamma_t \, \frac{A^{(t)}(z) \left[ \frac{\mathrm{d} A^{(t)}(z)}{\mathrm{d} z} + z \, \frac{\mathrm{d}^2 A^{(t)}(z)}{\mathrm{d} z^2} \right] - z \left[ \frac{\mathrm{d} A^{(t)}(z)}{\mathrm{d} z} \right] ^2}{[A^{(t)}(z)]^2} \, .
\end{align*}
The denominator of the fraction in each term in the sum is strictly positive for all $z > 0$. The numerator of the fraction in term $t \in I_c$ in the sum may be expanded as
\begin{align*}
\, & (1+\sum_{v=r_t}^{s_t} A^{(t)}_v z^v) (\sum_{u=r_t}^{s_t} u A^{(t)}_u z^{u-1} + \sum_{u=r_t}^{s_t} u(u-1) A^{(t)}_u z^{u-1})\\ \, & \phantom{-}\, - z (\sum_{u=r_t}^{s_t} u A^{(t)}_u z^{u-1}) (\sum_{v=r_t}^{s_t} v A^{(t)}_v z^{v-1}) \\
\, & = \sum_{u=r_t}^{s_t} u^2 A^{(t)}_u z^{u-1} + \sum_{u=r_t}^{s_t}\sum_{v=r_t}^{s_t} u(u-v)
A^{(t)}_u A^{(t)}_v z^{u+v-1} \, .
\end{align*}
Observe that in this expression, each term in the second summation with $u=v$ is zero, while each $(u,v)$ term in the second summation (with $u > v$) added to the corresponding $(v,u)$ term is positive for $z>0$, since $u(u-v) A^{(t)}_u A^{(t)}_v z^{u+v-1} + v(v-u) A^{(t)}_u A^{(t)}_v z^{u+v-1} = (u-v)^2 A^{(t)}_u A^{(t)}_v z^{u+v-1}> 0$ and therefore the second summation is nonnegative for $z>0$. Since the first summation is strictly positive for $z>0$, it follows that $\sff'(z)>0$ for all $z > 0$. 
\end{proof}

Note that, due to Lemma 1, the inverse of $\sff$, denoted by $\sff^{-1}:[0,M)\rightarrow\mathbb{R}^+$, is well-defined.


\section{Spectral Shape of Check-Hybrid GLDPC Codes}\label{section:spectral_shape_formula}
We next state and prove an expression for the spectral shape of check-hybrid GLDPC codes.
\begin{theorem}[Spectral shape of check-hybrid GLDPC codes]
Consider a GLDPC code ensemble with a regular VN set, composed of repetition codes all of length $q$, and a hybrid CN set, composed of a mixture of $n_c$ different linear block code types. Then, the weight spectral shape of the ensemble is given~by
\begin{align}\label{eq:G(alpha)_irregular_CN_set}
G(\alpha) & = (1-q) h(\alpha) -q\, \alpha \log \sff^{-1}(\alpha) \notag \\ 
\, & \phantom{------}+ q \left(\int\!\rho\right) \sum_{t\in I_c} \gamma_t \log A^{(t)}(\sff^{-1}(\alpha))
\end{align}
where $h(\alpha)=-\alpha\log \alpha -(1-\alpha)\log(1-\alpha)$ denotes the binary entropy function. 
\label{thm:main_result}
\end{theorem}
\begin{proof} 
In \cite{Flanagan09:weight}, a polynomial system solution for the spectral shape of irregular D-GLDPC codes was derived. In the special case of D-GLDPC codes where all VNs are repetition codes of length $q$, this is given by (special case of equation (8) in \cite{Flanagan09:weight})
\begin{align}
G(\alpha) & = \log B(x_{0},y_{0}) - \alpha \log x_{0} \notag \\ 
\, & \phantom{--} + q \left(\int\! \rho \right) \sum_{t \in I_c} \gamma_t \log A^{(t)}(z_0) + q
\log(1 - \frac{\beta}{q})
\label{eq:growth_rate_polynomial_Tanner}
\end{align}
where the values of $x_0$, $y_0$, $z_0$, $\beta$ in \eqref{eq:growth_rate_polynomial_Tanner} are found by solving the $(4 \times 4)$ polynomial system
\begin{equation} 
\left( \int\! \rho \right) \sum_{t \in I_c} \gamma_t \frac{z_0 \, \frac{\mathrm{d} A^{(t)}(z_0)}{\mathrm{d} z}}{A^{(t)}(z_0)} = \frac{\beta}{q} \; ,
\label{eq:Tanner_eq1}
\end{equation}
\begin{equation} 
\frac{x_0 y_0^q}{1 + x_0 y_0^q} = \alpha \; ,
\label{eq:Tanner_eq2}
\end{equation}
\begin{equation} 
\frac{x_0 y_0^q}{1 + x_0 y_0^q} = \frac{\beta}{q} \; ,
\label{eq:Tanner_eq3}
\end{equation}
and
\begin{equation} 
\frac{z_0 y_0}{1 + z_0 y_0} = \frac{\beta}{q} \; .
\label{eq:Tanner_eq4}
\end{equation}
Note that we are certain of the existence of a unique real solution to the polynomial system such that $x_0>0$, $y_0>0$, $z_0>0$, $\beta>0$, due to Hayman's formula. We solve this system of equations sequentially for the variables $\beta$, $z_0$, $y_0$ and $x_0$ (respectively). First, combining (\ref{eq:Tanner_eq2}) and (\ref{eq:Tanner_eq3}) yields 
\begin{equation} 
\beta = q \alpha \; .
\label{eq:Tanner_eq_beta}
\end{equation}
Substituting (\ref{eq:Tanner_eq_beta}) into (\ref{eq:Tanner_eq1}) yields $\sff(z_0) = \alpha$ which may be written as 
\begin{equation} 
z_0 = \sff^{-1} (\alpha) \; .
\label{eq:Tanner_eq_z0}
\end{equation}
Using (\ref{eq:Tanner_eq_beta}) and (\ref{eq:Tanner_eq_z0}) in (\ref{eq:Tanner_eq4}) yields 
\begin{equation} 
y_0 = \frac{\alpha}{(1-\alpha) \sff^{-1}(\alpha)} \; .
\label{eq:Tanner_eq_y0}
\end{equation}
Finally, substituting (\ref{eq:Tanner_eq_beta}) and (\ref{eq:Tanner_eq_y0}) into (\ref{eq:Tanner_eq3}) yields
\begin{equation} 
x_0 = \left( \frac{\alpha}{1-\alpha} \right)^{1-q} \left( \sff^{-1} (\alpha) \right)^q \; .
\label{eq:Tanner_eq_x0}
\end{equation}
Substituting (\ref{eq:Tanner_eq_beta}), (\ref{eq:Tanner_eq_z0}), (\ref{eq:Tanner_eq_y0}) and (\ref{eq:Tanner_eq_x0}) into (\ref{eq:growth_rate_polynomial_Tanner}), and simplifying, leads to \eqref{eq:G(alpha)_irregular_CN_set}. 
\end{proof}

The expression \eqref{eq:G(alpha)_irregular_CN_set} holds regardless of whether the ensemble has
good or bad spectral shape behavior. Note that, according to \eqref{eq:G(alpha)_irregular_CN_set},
the growth rate $G(\alpha)$ is well-defined only for $\alpha \in [0,M]$. This is as expected due to
the following reasoning. A codeword of weight $\alpha n$ naturally induces a distribution of bits on
the Tanner graph edges, $\alpha n q$ of which are equal to $1$. Also note that the maximum number of
ones in this distribution occurs when a maximum weight local codeword is activated for each of the
$\gamma_t m$ CNs of type $t \in I_c$, and is thus given by $m \sum_{t \in I_c} \gamma_t \bar{u}_t$.
Hence, we have $\alpha n q \leq m \sum_{t \in I_c} \gamma_t \bar{u}_t$, i.e., $\alpha \leq M$.

In Appendix~\ref{appendix:asymptotics} it is shown how, for small relative weight~$\alpha$, \eqref{eq:G(alpha)_irregular_CN_set} simplifies to a known expression that was derived in \cite{Tillich04:weight} for Tanner codes, and extended in \cite{paolini08:weight} to irregular GLDPC codes.

By considering Theorem \ref{thm:main_result} in the special case of Tanner codes, we obtain the following corollary.
\begin{corollary}[Spectral shape of Tanner codes]\label{theorem:G(alpha)_tanner_codes}
Consider a Tanner code ensemble where all variable component codes are length-$q$ repetition codes and where all check component codes are length-$s$ codes with weight enumerating function $A(z)=1+\sum_{u=r}^s A_u z^u$. The weight spectral shape of this ensemble is given by
\begin{equation}\label{eq:G(alpha)_tanner_codes}
G(\alpha) = (1-q) h(\alpha) - q\,\alpha\,\log(\sff^{-1}(\alpha)) + \frac{q}{s}\log A(\sff^{-1}(\alpha)) 
\end{equation}
\noindent where the function $\sff$ is given by (special case of (\ref{eq:f}))
\begin{align}\label{eq:f_regular_case}
\sff(z) = \frac{z\, A'(z)}{s\, A(z)} \, ,
\end{align}
and $\sff^{-1}:[0,M)\rightarrow\mathbb{R}^+$ is well-defined, where $M=\frac{\bu}{s}$.
\end{corollary}
Note that, in the special case where all CNs are SPC codes, \eqref{eq:G(alpha)_tanner_codes} becomes equal to the spectral shape expression for regular LDPC codes developed in \cite[Theorem~2]{orlitsky05:stopping} for the case of stopping sets. Also note that, in some cases, \eqref{eq:G(alpha)_tanner_codes} can be expressed analytically as $\sff^{-1}(\alpha)$ admits an analytical form. An example is given in Appendix~\ref{appendix:closed_form}.


\section{Symmetry of the Weight Spectral Shape}
Consider a GLDPC code ensemble with a regular VN set and a hybrid CN set. In this section, we show how a symmetry in the overall weight spectral shape of the ensemble is induced by local symmetry properties in the WEFs of the CNs.

\begin{definition}
For CN type $t \in I_c$, let $U_t = \{u \in \mathbb{N} | A^{(t)}_u > 0\}$. Then, we define
\begin{align}\label{eq:U*_definition}
U^{(t*)} = \{v\in\mathbb{N} | \bar{u}_t-v \in U_t \}\, .
\end{align}
\end{definition}
Note that for any CN type, we always have $0 \in U^{(t*)}$ and $\bar{u}_t \in U^{(t*)}$.

\begin{definition}\label{def:A_symmetry}
The WEF of CN type $t \in I_c$ is said to be \emph{symmetric} if and only if $A^{(t)}_{\bar{u}_t-u}=A^{(t)}_u$ for all $u \in U_t$.
\end{definition}
Note that if the WEF of CN type $t \in I_c$ is symmetric, then we have $U^{(t*)}=U^{(t)}$.

\begin{lemma}\label{lemma:all_1_codeword}
The WEF of CN type $t \in I_c$ is symmetric if and only if the all-$1$ codeword belongs to this code.
\end{lemma}
The proof of Lemma~\ref{lemma:all_1_codeword} is omitted due to space constraints.\footnote{The
sufficient condition (if a linear block code has the all-$1$ codeword then its WEF is symmetric) is
a well-known result in classical coding theory. On the other hand, we proved the necessary condition
by assuming that CN type $t \in I_c$ has a symmetric WEF and a maximum codeword weight $D<s_t$, and
by showing that these assumptions lead to a contradiction. As pointed out by one of the anonymous
reviewers, the necessary condition may also be proved by reasoning on the expected weight of a
randomly selected codeword in a linear block code, under the symmetry hypothesis.}
\begin{lemma}\label{lemma:A_symmetry}
The WEF of CN type $t \in I_c$ fulfills
\begin{equation}
A^{(t)}(z) = z^{\bar{u}_t} A^{(t)}\left(z^{-1}\right) \; ,
\label{eq:symmetry_of_CN_type} 
\end{equation} 
for all $z \in \mathbb{R}^+$ if and only if it is symmetric.
\end{lemma}
\begin{proof}
We have
\begin{align*}
A^{(t)}\left(z^{-1}\right) = \frac{\sum_{u\in U_t} A^{(t)}_u z^{\bar{u}_t-u}}{z^{\bar{u}_t}} =
\frac{\sum_{v\in U^{(t*)}} A^{(t)}_{\bar{u}_t-v} z^{v}}{z^{\bar{u}_t}}
\end{align*}
where the final equality is obtained by $v=\bar{u}_t-u$. The proof is completed by observing that, if and only if $A^{(t)}(z)$ is symmetric, we have $U^{(t*)}=U_{t}$ and $A^{(t)}_{\bar{u}_t-v}=A^{(t)}_v$ for all $v \in U_t$.
\end{proof}

\begin{lemma}\label{lemma:f_symmetry}
The function $\sff$ defined by (\ref{eq:f}) fulfills
\begin{align}\label{eq:f_symmetry}
\sff(z) = M - \sff\left(z^{-1}\right)
\end{align}
$\forall$ $z \in \mathbb{R}^+$ if and only if $A^{(t)}(z)$ is symmetric for every $t \in I_c$.
\end{lemma}
\begin{proof}
First we note that if and only if the WEF of CN type $t \in I_c$ is symmetric, we have, differentiating (\ref{eq:symmetry_of_CN_type}),
\[
\frac{\mathrm{d} A^{(t)}(z)}{\mathrm{d} z} = - z^{\bar{u}_t-2} \frac{\mathrm{d} A^{(t)}(z^{-1})}{\mathrm{d} z^{-1}} + \bar{u}_t z^{\bar{u}_t-1} A^{(t)}(z^{-1}) \, .
\]
Multiplying by $z$ and using (\ref{eq:symmetry_of_CN_type}) yields
\begin{equation}
z \frac{\mathrm{d} A^{(t)}(z)}{\mathrm{d} z} = - z^{\bar{u}_t-1} \frac{\mathrm{d} A^{(t)}(z^{-1})}{\mathrm{d} z^{-1}} + \bar{u}_t A^{(t)}(z) \; .
\label{eq:symmetry_of_CN_type_derivative}
\end{equation}
Then,
\begin{align*}
M - \sff\left(z^{-1} \right) & = M - \left( \int\! \rho \right) \sum_{t \in I_c} \gamma_t \left( \frac{z^{-1} \frac{\mathrm{d} A^{(t)}(z^{-1})}{\mathrm{d} z^{-1}}}{A^{(t)}(z^{-1})} \right) \\
\, & \!\!\!\!\!\!\!\!\!\!\!\!\!\!\!\!\!\!\!\!\!\!\!\!\!\!\!\!\!\!\!\!\stackrel{\textrm{(a)}}{=}
\left( \int\! \rho \right) \sum_{t \in I_c} \gamma_t \left( \bu_t - \frac{z^{\bar{u}_t-1}
\frac{\mathrm{d} A^{(t)}(z^{-1})}{\mathrm{d} z^{-1}}}{A^{(t)}(z)} \right) \\
\, & \!\!\!\!\!\!\!\!\!\!\!\!\!\!\!\!\!\!\!\!\!\!\!\!\!\!\!\!\!\!\!\!\stackrel{\textrm{(b)}}{=} \left( \int\! \rho \right) \sum_{t \in I_c} \gamma_t \frac{\bu_t A^{(t)}(z) + z \frac{\mathrm{d} A^{(t)}(z)}{\mathrm{d} z} - \bar{u}_t A^{(t)}(z)}{A^{(t)}(z)}=\sff(z)
\end{align*}
where we have used \eqref{eq:M_definition} and \eqref{eq:symmetry_of_CN_type} in (a), and
\eqref{eq:symmetry_of_CN_type_derivative} in (b).
\end{proof}

\begin{lemma}\label{lemma:finv_symmetry}
The inverse function $\sff^{-1}$ fulfills
\begin{align}\label{eq:finv_symmetry}
\sff^{-1}(M-\alpha)=\frac{1}{\sff^{-1}(\alpha)}
\end{align}
$\forall$ $\alpha \in (0,1)$ if and only if $A^{(t)}(z)$ is symmetric for every \mbox{$t \in I_c$}.
\end{lemma}
\begin{proof}
By Lemma~\ref{lemma:f_symmetry}, the function $\sff$ fulfills \eqref{eq:f_symmetry} if and only if $A^{(t)}(z)$ is symmetric for every $t \in I_c$. By applying the inverse function to both sides of \eqref{eq:f_symmetry} and by letting $\sff(z^{-1})=\alpha$ for all $z \in \mathbb{R}^+ \backslash \{0\}$, we obtain the statement.
\end{proof}

\begin{theorem}[Sufficient condition for symmetry]\label{theorem:G_symmetry_sufficient} Consider a GLDPC code ensemble with a regular VN set, composed of repetition codes all of length $q$, and a hybrid CN set, composed of a mixture of $n_c$ different linear block codes. If $A^{(t)}(z)$ is symmetric for each $t \in I_c$, then the spectral shape of the ensemble fulfills
\begin{align}\label{eq:G_symmetry}
G(M-\alpha) = G(\alpha)
\end{align}
for all $\alpha \in (0,1)$.
\end{theorem}
\begin{proof}
By Lemma \ref{lemma:all_1_codeword}, the hypothesis that $A^{(t)}(z)$ is symmetric for each $t \in I_c$ implies $\bar{u}_t = s_t$ for every $t \in I_c$, and therefore $M=1$. From \eqref{eq:G(alpha)_irregular_CN_set} we have:
\begin{align*}
\, & G(M-\alpha) = (1-q) h(M-\alpha) - q(M-\alpha) \log \sff^{-1}(M-\alpha) \\ 
\, & \phantom{--------} + q \left(\int\! \rho\right) \sum_{t \in I_c} \gamma_t \log A^{(t)} \left(\sff^{-1}(M-\alpha)\right) \\
\, & \stackrel{\textrm{(a)}}{=} (1-q) h(M-\alpha) - q(M-\alpha) \log \frac{1}{\sff^{-1}(\alpha)}  \\ 
\, & \phantom{--------} + q \left(\int\! \rho\right) \sum_{t \in I_c} \gamma_t \log A^{(t)} \left(\frac{1}{\sff^{-1}(\alpha)}\right) \\
\, & \stackrel{\textrm{(b)}}{=} (1-q) h(M-\alpha) - q\, \alpha\, \log(\sff^{-1}(\alpha)) \\ 
\, & \phantom{----------} + q \left(\int\! \rho\right) \sum_{t \in I_c} \gamma_t \log A^{(t)}
(\sff^{-1}(\alpha)) \\
\, & = G(\alpha)
\end{align*}
where (a) follows from symmetry of $A^{(t)}(z)$ and Lemma~\ref{lemma:finv_symmetry}, (b) from symmetry of $A^{(t)}(z)$, Lemma~\ref{lemma:A_symmetry} and \eqref{eq:M_definition}, and the final line from $M=1$.
\end{proof}
\begin{theorem}[Necessary condition for symmetry]\label{theorem:G_symmetry_necessary} Consider a GLDPC code ensemble with a regular VN set, composed of repetition codes all of length $q$, and a hybrid CN set, composed of a mixture of $n_c$ different linear block codes. If the spectral shape fulfills $G(M-\alpha) = G(\alpha)$ for all $\alpha \in (0,1)$, then $M$ is a fixed point of the function
\begin{equation}\label{eq:Gamma}
\Gamma(x)=2\,\, \sff\left( \left(\frac{x}{2-x}\right)^{\frac{q-1}{q}} \right)\, .
\end{equation}
\end{theorem}
\begin{proof}
The proof is somewhat lengthy and we only sketch it due to space constraints. The first step consists of developing an expression for $F(\alpha)=G(M-\alpha)-G(\alpha)$. The obtained expression can be written as a function of $z$, denoted by $J(z)$, by defining $z=\sff^{-1}(\alpha)$. We must have $J(z)=0$ for all $z \in (0,+\infty)$, and therefore $J'(z)=0$ for all $z \in (0,+\infty)$. This latter condition must hold in particular for $z=\sff^{-1}(\frac{M}{2})$, in which case we obtain the compact condition
$$
q\, \log\sff^{-1}\left(\frac{M}{2}\right)=(1-q)\log\left(\frac{2-M}{M}\right)
$$
which is equivalent to
$$
M=2\,\, \sff\left( \left(\frac{M}{2-M}\right)^{\frac{q-1}{q}} \right)\, .
$$
\end{proof}

Note that if $A^{(t)}(z)$ is symmetric for each $t \in I_c$ then we have $M=1$, which is always a fixed point of $\Gamma(x)$ defined in \eqref{eq:Gamma} since it is possible to show that $\sff(1)=\frac{1}{2}$.


\section{Examples}
\begin{figure}[t]
\begin{center}
\psfragscanon
\psfrag{xaxis}[lt]{\scriptsize{$\alpha$}}
\psfrag{yaxis}[b]{\scriptsize{\textsf{Growth rate,} $G(\alpha)$}}
\includegraphics[%
  width=0.8\columnwidth,
  keepaspectratio]{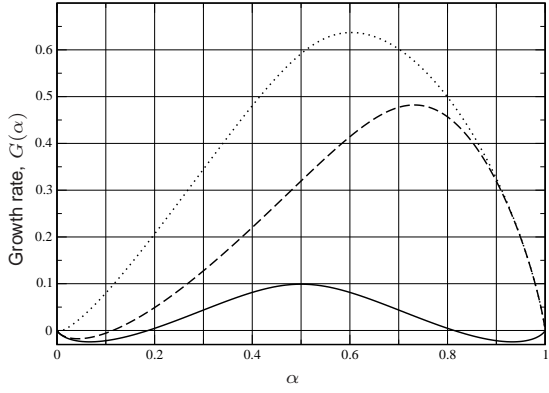}
\end{center}
\caption{Spectral shapes of the Tanner code ensemble in Example~\ref{example:1}. Solid: weight spectral shape (relative minimum distance: $\alpha^* = 0.18650$). Dashed: stopping set size spectral shape under MAP decoding at the CNs (relative minimum stopping set size: $\alpha_{\Phi}^* = 0.11414$). Dotted: stopping set size spectral shape under BD decoding at the CNs (relative minimum stopping set size: $\alpha_{\Psi}^* = 0.01025$).}
\label{fig:growth_rate_rep2_H7}
\end{figure}
\begin{figure}[t]
\begin{center}
\psfragscanon
\psfrag{xaxis}[lt]{\scriptsize{$\alpha$}}
\psfrag{yaxis}[b]{\scriptsize{\textsf{Growth rate,} $G(\alpha)$}}
\includegraphics[%
  width=0.8\columnwidth,
  keepaspectratio]{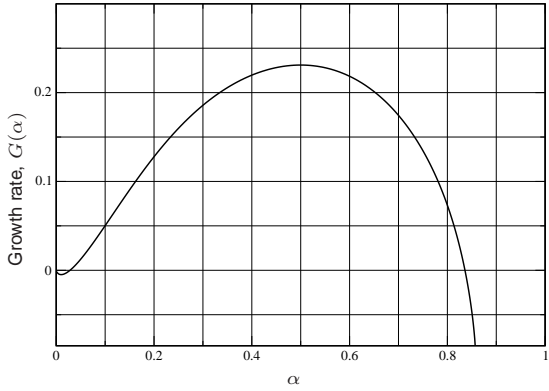}
\end{center}
\caption{Weight spectral shape of the check-hybrid GLDPC code ensemble in Example~\ref{example:2}. Relative minimum distance: $\alpha^* = 0.028179$.}
\label{fig:growth_rate_check_hybrid}
\end{figure}
\begin{figure}[t]
\begin{center}
\psfragscanon
\psfrag{xaxis}[lt]{\scriptsize{$\alpha$}}
\psfrag{yaxis}[b]{\scriptsize{\textsf{Growth rate,} $G(\alpha)$}}
\includegraphics[%
  width=0.8\columnwidth,
  keepaspectratio]{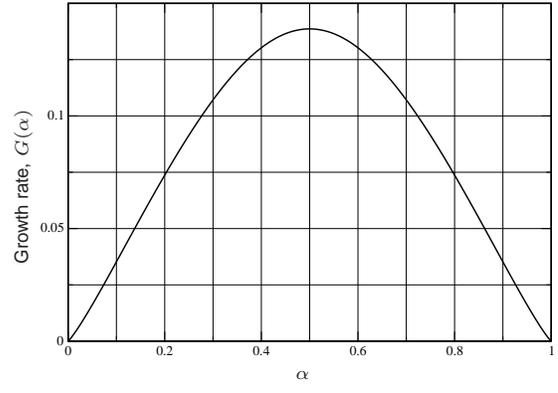}
\end{center}
\caption{Weight spectral shape of the Tanner code ensemble in Example~\ref{example:3}.}
\label{fig:growth_rate_asymp_bad}
\end{figure}
\begin{example}[Tanner code with $(7,4)$ Hamming CNs]\label{example:1} Consider a rate $R=1/7$
Tanner code ensemble where all VNs have degree $2$ and where all CNs are $(7,4)$ Hamming codes
(it was shown in \cite{lentmaier99:GLD,boutros99:GLD} that this ensemble has good spectral shape
behavior). The WEF of a Hamming $(7,4)$ CN is given by $A(z)=1 + 7 z^3 + 7 z^4 + z^7$, while its
local MAP-SSEF and BD-SSEF are given by $\Phi(z)=1+7 z^3+10 z^4+21 z^5+7 z^6+z^7$ and $\Psi(z)=1+35
z^3+35 z^4+21 z^5+7 z^6+z^7$ respectively. Note that we have $M=\frac{\bu}{s}=1$ in all three cases.
A plot of $G(\alpha)$, $G_{\Phi}(\alpha)$ and $G_{\Psi}(\alpha)$ obtained by implementation of
\eqref{eq:G(alpha)_tanner_codes} is depicted in Fig.~\ref{fig:growth_rate_rep2_H7}. We observe that
$A(z)$ satisfies the conditions of Theorem~\ref{theorem:G_symmetry_sufficient}. This is reflected by
the fact that the weight spectral shape $G(\alpha)$ is symmetric with respect to $\alpha=1/2$.
\end{example}
\begin{example}[Check-hybrid ensemble]\label{example:2} Consider a rate $R=1/3$ check-hybrid GLDPC code ensemble where all VNs are repetition codes of length $q=3$ and whose CN set is composed of a mixture of two linear block code types ($I_c=\{1,2\}$). CNs of type $1 \in I_c$ are length-$7$ SPC codes with WEF $A^{(1)}(z)=[(1+z)^7+(1-z)^7]/2$ and $\gamma_1=0.722$, while CNs of type $2 \in I_c$ are $(7,4)$ codes with WEF $A^{(2)}(z)=1+5 z^2+7 z^4+3 z^6$ and $\gamma_2=0.278$. The weight spectral shape of this ensemble, obtained from \eqref{eq:G(alpha)_irregular_CN_set}, is depicted in Fig.~\ref{fig:growth_rate_check_hybrid}. Note that for this ensemble, $M=6/7$. This value is not a fixed point of the function $F(x)$ defined in \eqref{eq:Gamma} (the only fixed point between $0$ and $1$ is $x=0.888421$). As expected, the weight spectral shape does not exhibit any symmetry property.
\end{example}
\begin{example}[Ensemble with bad spectral shape behavior]\label{example:3}
Consider a rate $R=1/5$ Tanner code ensemble where all VNs are repetition codes of length $q=2$ and
and where all CNs are $(5,3)$ linear block codes with WEF $A(z)=1+3 z^2+3 z^3+z^5$. This ensemble is
known to have bad spectral shape behavior ($\alpha^*=0$) since we have $\lambda'(0)C=6/5>1$, where
$\lambda(x)=x$ and $C=2 A_2 / s$ \cite{Tillich04:weight,paolini08:weight}. A plot of the weight
spectrum for this ensemble, obtained from \eqref{eq:G(alpha)_tanner_codes} is depicted in
Fig.~\ref{fig:growth_rate_asymp_bad}. We observe that the plot of $G(\alpha)$ is symmetric, due to
the fact that $A(z)$ is symmetric ($M=1$). As expected, the derivative of $G(\alpha)$ at $\alpha=0$
is positive and hence $\alpha^*=0$.
\end{example}


\section{Conclusion}
A simple expression has been developed for both the weight and the stopping set size spectral shape
of GLDPC code ensembles with a regular VN set and a hybrid CN set. Some known results (specifically,
an expression for the spectral shape of regular LDPC codes and an asymptotic expression of the
spectral shape for GLDPC codes as the normalized weight tends to zero) follow as corollaries of the
developed formula. Moreover, symmetry properties of the spectral shape function have been discussed.
A sufficient condition and a necessary condition for the weight spectral shape function to be
symmetric have been identified.
\appendices


\section{Asymptotic Case $\alpha \rightarrow 0$}\label{appendix:asymptotics}
For small $\alpha$, the expression \eqref{eq:G(alpha)_irregular_CN_set} reduces to a known formula first developed in \cite{Tillich04:weight} for Tanner codes, and extended in \cite{paolini08:weight} to irregular GLDPC codes. This formula is here obtained as a simple corollary of Theorem~\ref{thm:main_result}.
\begin{corollary}
In the limit where $\alpha \rightarrow 0$, the growth rate of the weight distribution of a GLDPC code ensemble with a regular VN set fulfills
\begin{align}\label{eq:G_asymptotic}
G(\alpha) \rightarrow & \left(q-\frac{q}{r}-1\right)\! \alpha\log\alpha +\! \frac{q\,\alpha}{r}\log \left(\! e r\! \int\! \rho \sum_{t \in X_c} \gamma_t A_{r}^{(t)} \right) \, . \notag \\
\end{align}
where $r$ denotes the smallest minimum distance of all CN types (i.e. $r \triangleq \min_{t \in I_c} r_t$), $X_c$ denotes the set of CN types with this minimum distance (i.e. $X_c = \{ t \in I_c \; : \; r_t = r \}$) and $e$ denotes Napier's number.
\end{corollary}
\begin{proof}
Let $\alpha = \sff(z_0)$. From the definition of $\sff$ given in \eqref{eq:f}, it is readily shown that if $\alpha \rightarrow 0$, we must have $z_0 \rightarrow 0$. Next, note that in (\ref{eq:f}), each term $t \in I_c$ in the sum on the right-hand side is a rational polynomial in $z_0$ whose denominator tends to unity as $\alpha \rightarrow 0$, and whose numerator is dominated as $\alpha \rightarrow 0$ by the term corresponding to the lowest power of $z_0$. Therefore, as $\alpha \rightarrow 0$, (\ref{eq:f}) becomes
$\alpha \rightarrow \left( \int \rho \right) \sum_{t \in X_c} \gamma_t r A_r^{(t)} z_0^r$, or equivalently,
\begin{align}\label{eq:limit_z0}
z_0 \rightarrow \left( \frac{\alpha}{r \int\! \rho \sum_{t \in X_c} \gamma_t A_r^{(t)} }\right)^{1/r}\, .
\end{align}
Using $\log (1+x) \rightarrow x$ as $x \rightarrow 0$, we have that for every $t \in I_c$
$\log A^{(t)}(z_0) \rightarrow A_{r_t}^{(t)} z_0^{r_t}$ %
as $\alpha \rightarrow 0$, and so
\begin{equation}\label{eq:sum_log_A(z)_limit}
\sum_{t \in I_c} \gamma_t \log A^{(t)}(z_0) \rightarrow \sum_{t \in X_c} \gamma_t A_{r}^{(t)} z_0^{r} = \frac{\alpha}{r \int\! \rho} \; . 
\end{equation}

Note that in these steps we have again used the fact that a polynomial expression in $z_0$ is
dominated by its lowest degree term as $z_0 \rightarrow 0$. Next, observe that $h(\alpha)
\rightarrow - \alpha \log \alpha$ as \mbox{$\alpha \rightarrow 0$} and therefore we obtain (using
\eqref{eq:limit_z0} and \eqref{eq:sum_log_A(z)_limit} in~\eqref{eq:G(alpha)_irregular_CN_set})
\begin{align*}G(\alpha) \rightarrow & (q-1)\,\alpha \log\alpha - \frac{q\,\alpha}{r}\log \left(
\frac{\alpha}{r \int \rho \sum_{t \in X_c} \gamma_t A_{r}^{(t)}} \right)\\
\, & + \left( q \int\! \rho \right) \left( \frac{\alpha}{r \int\! \rho} \right)\end{align*} %
which coincides with \eqref{eq:G_asymptotic}.
\end{proof}


\section{Closed Form Expressions for the Growth Rate}\label{appendix:closed_form}
It is worthwhile to note that in some cases, \eqref{eq:G(alpha)_tanner_codes} can be expressed in
closed form because $\sff^{-1}(\alpha)$ can be expressed analytically. This is the case, for
instance, for the $(3,6)$ regular LDPC ensemble, for which $\sff(z) = \alpha$ becomes
$ax^3+bx^2+cx+d=0$, where $x=z^2$ and $(a,b,c,d) = (\alpha-1, 15\alpha-10,15\alpha-5,\alpha)$. This
cubic equation in $x$ may be solved by Cardano's method (see, e.g., \cite[p.
17]{wikipedia:Cardanos_method}; the discriminant $\Delta = \rho^3 + \mu^2$ is negative for every
$\alpha \in (0,1)$, where 
\[
\rho = \frac{3 a c - b^2}{9 a^2} \; ; \; \mu = \frac{9 a b c - 27a^2 d - 2 b^3}{54 a^3} \; .
\]
The required solution is then uniquely and analytically identified as $\sff^{-1}(\alpha) = z =
\sqrt{x}$ where $x = 2 \sqrt{-\rho} \cos \left( \theta/3 \right) - \frac{b}{3a} > 0$ and $\theta =
\tan^{-1} \left( \sqrt{-\Delta}/\mu \right)$.

\newpage
Similarly, the weight spectral shape of a $(4,8)$ regular LDPC ensemble may be expressed in
closed form through the solution of a quartic equation.


\bibliographystyle{IEEEtran}

\end{document}